%% file: CDC2015.tex
\documentclass[letterpaper, 10pt, conference]{ieeeconf}  

\include{before_document}

\begin{document}

\title{Bearing-Based Formation Stabilization with Directed Interaction Topologies}
\author{Shiyu Zhao and Daniel Zelazo
\thanks{S. Zhao is with the Department of Mechanical Engineering, University of California, Riverside, USA. {\tt\small shiyuzhao@engr.ucr.edu}}
\thanks{D. Zelazo is with the Faculty of Aerospace Engineering, Technion - Israel Institute of Technology, Haifa, Israel. {\tt\small dzelazo@technion.ac.il}}
}
\IEEEoverridecommandlockouts
\maketitle
\begin{abstract}
This paper studies the problem of stabilizing target formations specified by inter-neighbor bearings with relative position measurements.
While the undirected case has been studied in the existing works, this paper focuses on the case where the interaction topology is \emph{directed}.
It is shown that a linear distributed control law, which was proposed previously for undirected cases, can still be applied to the directed case.
The formation stability in the directed case, however, relies on a new notion termed \emph{bearing persistence}, which describes whether or not the directed underlying graph is persistent with the bearing rigidity of a formation.
If a target formation is not bearing persistent, undesired equilibriums will appear and global formation stability cannot be guaranteed.
The notion of bearing persistence is defined by the bearing Laplacian matrix and illustrated by simulation examples.
\end{abstract}
\overrideIEEEmargins

\section{Introduction}

In recent years, there has been a growing research interest on bearing-based formation control \cite{bishop2010SCL,bishopconf2011rigid,Eren2012IJC,Eric2014ACC,zhao2014TACBearing,zhao2015ECC,zhao2015MSC,zelazo2015CDC} and bearing-based network localization \cite{BishopTAES2009,Piovan2013Automatica,Shames2013TAC,ZhuGuangwei2014Automatica,zhao2015NetLocalization}.
The bearing-based approaches provide solutions to control or localize networks in arbitrary dimensional spaces merely using the inter-neighbor bearing information.
As shown in \cite{zhao2015MSC}, the bearing-based approach also provides a simple solution to formation scale control, which is complicated to solve using the distance or relative position based approaches \cite{Coogan2012Scale,Park2014IJRNC}.
While most of the existing works focused on the \emph{undirected} case, this paper focuses on bearing-based formation control with \emph{directed} interaction topologies.

The specific problem considered in this paper is described below.
Suppose the target formation is static and its desired geometric pattern is defined by the inter-neighbor bearing constraints. Each agent can obtain the relative positions of its neighbors via wireless communication.
The control objective is to steer the formation from an initial configuration to a final configuration that satisfies the inter-neighbor bearing constraints.
The linear formation control law considered in this paper is the one proposed in our previous work \cite{zhao2015ECC} for formations with undirected interaction topologies.
It will be shown in this paper that the control law can also be applied to the directed case, but the formation stability in the directed case is more complicated to analyze.

In the undirected case, the formation stability merely relies on the bearing rigidity properties of the target formation \cite{zhao2015ECC}.
The bearing rigidity of networks or formations has been studied in \cite{bishopconf2011rigid,Eren2012IJC,zhao2014TACBearing,zelazo2015CDC}.
It is shown in \cite{zhao2014TACBearing} that the geometric pattern of a target formation can be uniquely determined by the inter-neighbor bearings if and only if the target formation is infinitesimally bearing rigid.
It is further shown in \cite{zhao2015ECC} that the global formation stability will be guaranteed if the target formation is infinitesimally bearing rigid.
However, in the directed case, the formation stability will rely on not only the bearing rigidity properties but also a new notion termed \emph{bearing persistence} as observed in this paper.
The notion of bearing persistence is defined based on a special matrix, termed the \emph{bearing Laplacian}, which plays important roles in the bearing-based formation control and network localization problems as shown in \cite{zhao2015ECC,zhao2015MSC,zhao2015NetLocalization}.
The bearing Laplacian of a formation can be viewed as a matrix-weighted graph Laplacian.
In the undirected case, the bearing Laplacian is symmetric positive semi-definite and its null space is exactly the same as that of the bearing rigidity matrix.
However, the bearing Laplacian may lose these elegant properties in the directed case.

The work presented in this paper is a first step towards solving the general problem of bearing-based formation control with directed interaction topologies.
The notion of bearing persistence is defined and explored.
Specifically, a formation is bearing persistent when the null spaces of the bearing Laplacian and the bearing rigidity matrix are the same.
We further show that if the formation is not bearing persistent, undesired equilibriums will appear and the global formation stability cannot be guaranteed.
We also present a sufficient condition to ensure bearing persistence for formations in the two-dimensional space.
The spectrum of the bearing Laplacian in the directed case, which determines the formation stability, is however not analyzed in this paper and will be addressed in the future.
Finally, since the bearing-based formation control problem is mathematically equivalent to the bearing-only network localization problem as shown in \cite{zhao2015ECC}, the results presented in this paper can be easily transferred to solve the bearing-only network localization problem with directed interaction topologies.

\paragraph*{Notations}
Given $A_i\in\mathbb{R}^{p\times q}$ for $i=1,\dots,n$, denote $ \mydiag(A_i)\triangleq\blkdiag\{A_1,\dots,A_n\}\in\mathbb{R}^{np\times nq}$.
Let $\Null(\cdot)$ and $\Range(\cdot)$ be the null space and range space of a matrix, respectively.
Denote $I_d\in\R^{d\times d}$ as the identity matrix, and $\one\triangleq[1,\dots,1]^\T$.
Let $\|\cdot\|$ be the Euclidian norm of a vector or the spectral norm of a matrix, and $\otimes$ be the Kronecker product.

\section{Preliminaries to Bearing Rigidity Theory}

Bearing rigidity theory plays a key role in the analysis of bearing-based distributed control and estimation problems.
In this section, we revisit some notions and results in the bearing rigidity theory \cite{zhao2014TACBearing}.
It is worth noting that the bearing rigidity theory presented in \cite{zhao2014TACBearing} is developed for the case of undirected graphs. However, it can be generalized to the directed case without any difficulties.

We first introduce an orthogonal projection matrix operator that will be used throughout the paper.
For any nonzero vector $x\in\R^d$ ($d\ge2$), define $P: \R^d\rightarrow\R^{d\times d}$ as
\begin{align*}
    P(x) \triangleq I_d - \frac{x}{\|x\|}\frac{x^\T }{\|x\|}.
\end{align*}
For notational simplicity, we denote $P_x=P(x)$.
The matrix $P_x$ geometrically projects any vector onto the orthogonal compliment of $x$.
It is easily verified that $P_x^\T =P_x$, $P_x^2=P_x$, and $P_x$ is positive semi-definite.
Moreover, $\Null(P_x)=\myspan\{x\}$ and the eigenvalues of $P_x$ are $\{0,1^{(d-1)}\}$.
Any two nonzero vectors $x,y\in\R^d$ are parallel if and only if $P_xy=0$ (or equivalently $P_yx=0$).

Given a finite collection of $n$ points $\{p_i\}_{i=1}^n$ in $\R^d$ ($n\ge2$, $d\ge2$), denote $p=[p_1^\T,\dots,p_n^\T]^\T\in\mathbb{R}^{dn}$.
A \emph{formation} in $\R^d$, denoted as $\G(p)$, is a directed graph $\G=(\V,\E)$ together with $p$, where vertex $i\in\V$ in the graph is mapped to the point $p_i$.
For a formation $\G(p)$, define the \emph{edge vector} and the \emph{bearing}, respectively, as
\begin{align*}
e_{ij}\triangleq p_j-p_i, \quad g_{ij}\triangleq e_{ij}/\|e_{ij}\|, \quad \forall(i,j)\in\E.
\end{align*}
The bearing $g_{ij}$ is a unit vector.

\begin{definition}[Bearing Equivalence]\label{definition_bearingEquivalence}
Two formations $\G(p)$ and $\G(p')$ are \emph{bearing equivalent} if $P_{g_{ij}}g'_{ij}=0$ for all $(i,j)\in\E$.
\end{definition}

By Definition~\ref{definition_bearingEquivalence}, bearing equivalent formations have parallel inter-neighbor bearings.
Suppose $|\E|=m$ and index all the directed edges from $1$ to $m$.
Re-express the edge vector and the bearing as $e_{k}$ and $g_{k}\triangleq {e_{k}}/{\|e_{k}\|}$, $\forall k\in\{1,\dots,m\}$.
Let $e=[e_1^\T ,\dots,e_m^\T ]^\T$ and $g=[g_1^\T ,\dots,g_m^\T ]^\T$.
Note $e$ satisfies $e=\bar{H}p$ where $\bar{H}=H\otimes I_d$ and $H$ is the incidence matrix of the graph $\mathcal{G}$ \cite{BookGodsilGraph}.
Define the \emph{bearing function} $F_B: \R^{dn}\rightarrow\R^{dm}$ as
\begin{align*}
    \FB(p)\triangleq [g_1^\T, \dots,g_m^\T]^\T.
\end{align*}
The bearing function describes all the bearings in the formation.
The \emph{bearing rigidity matrix} is defined as the Jacobian of the bearing function,
\begin{align}\label{eq_rigidityMatrixDefinition}
    \RB(p) \triangleq \frac{\partial \FB(p)}{\partial p}\in\R^{dm\times dn}.
\end{align}
Let $\delta p$ be a variation of $p$.
If $\RB(p)\delta p=0$, then $\delta p$ is called an \emph{infinitesimal bearing motion} of $\G(p)$.
Any formation always has two kinds of \emph{trivial} infinitesimal bearing motions: translation and scaling of the entire formation.
We next define one of the most important concepts in bearing rigidity theory.

\begin{definition}[{Infinitesimal Bearing Rigidity}]\label{definition_infinitesimalParallelRigid}
    A formation is \emph{infinitesimally bearing rigid} if all the infinitesimal bearing motions of the formation are trivial.
\end{definition}

Some useful results are listed as below; for proofs, please refer to \cite{zhao2014TACBearing}.

\begin{lemma}[\cite{zhao2014TACBearing}]\label{lemma_BearingRigidityProperty}
For any formation $\G(p)$, the bearing rigidity matrix defined in \eqref{eq_rigidityMatrixDefinition} satisfies
\begin{enumerate}[(a)]
\item $\RB(p)= \mydiag\left({P_{g_k}}/{\|e_k\|}\right)\bar{H}$;
\item $\rank(\RB)\le dn-d-1$ and $\myspan\{\one\otimes I_d, p\}\subseteq \Null(\RB)$.
\end{enumerate}
\end{lemma}

\begin{theorem}[\cite{zhao2014TACBearing}]\label{result_bearingEquivalentCondition}
Any two formations $\G(p)$ and $\G(p')$ are bearing equivalent if and only if $\RB(p)p'=0$.
\end{theorem}

\begin{theorem}[\cite{zhao2014TACBearing}]\label{theorem_conditionInfiParaRigid}
    For any formation $\G(p)$, the following statements are equivalent:
    \begin{enumerate}[(a)]
    \item $\G(p)$ is {infinitesimally bearing rigid};
    \item $\G(p)$ can be uniquely determined up to a translation and a scaling factor by the inter-neighbor bearings;
    \item $\rank(\RB)=dn-d-1$;
    \item $\Null(\RB)=\myspan\{\one\otimes I_d, p\}$.
    \end{enumerate}
\end{theorem}

\section{Problem Formulation}

In this section, we first state the problem of bearing-based formation stabilization, and then present a linear bearing-based formation control law.

\subsection{Problem Statement}

Consider a formation of $n$ agents in $\R^d$ ($n\ge2$, $d\ge2$).
Denote ${p}_i\in\R^d$ as the position of agent $i\in\{1,\dots,n\}$, and let $p=[p_1^\T ,\dots,p_n^\T ]^\T \in\R^{dn}$.
The dynamics of each agent is a single integrator, $\dot{p}_i(t)=u_i(t)$, where $u_i(t)\in\R^d$ is the input to be designed.
The underlying graph $\G=(\V,\E)$ is directed, fixed, and connected.
If $(i,j)\in\E$, we say agent $i$ can ``see'' agent $j$, which means agent $i$ can access the relative position of agent $j$.
We call $\G(p)$ as a {directed} (respectively, undirected) formation if $\G$ is directed (respectively, undirected).

The aim of the control is to stabilize a target formation specified by constant bearing constraints $\{g_{ij}^*\}_{(i,j)\in\E}$.
The bearing constraints must be {feasible} such that there exist formations satisfying these constraints.
The problem to be solved in this paper is stated as below.

\begin{problem}\label{problem_bearingonlyformationcontrol_BE}
    Given feasible constant bearing constraints $\{g_{ij}^*\}_{(i,j)\in\E}$ and an initial configuration $p(0)$, design $u_i(t)$ ($i\in\V$) based on relative position measurements $\{p_i(t)-p_j(t)\}_{j\in\N_i}$ such that $P_{g_{ij}^*}g_{ij}(t)\rightarrow 0$ as $t\rightarrow\infty$ for all $(i,j)\in\E$.
\end{problem}

Let $\G(p^*)$ be an arbitrary formation that satisfies the bearing constraints $\{g_{ij}^*\}_{(i,j)\in\E}$.
Problem~\ref{problem_bearingonlyformationcontrol_BE} actually requires that the formation $\G(p(t))$ converges to a final formation that is \emph{bearing equivalent} to $\G(p^*)$.
This means that we allow the bearings $g_{ij}(t)$ to converge to either $g_{ij}^*$ or $-g_{ij}^*$.  This is different from the problem setup in \cite{zhao2014TACBearing} where $g_{ij}(t)$ must converge to $g_{ij}^*$.
An illustrative example is shown in Figure~\ref{fig_twoKindsofTargetFormation}.
The formation in Figure~\ref{fig_twoKindsofTargetFormation}(a) has the bearings as $\{g^*_{ij}\}_{(i,j)\in\E}$, whereas the formation in Figure~\ref{fig_twoKindsofTargetFormation}(b) has the opposite bearings as $\{-g^*_{ij}\}_{(i,j)\in\E}$.
Problem~\ref{problem_bearingonlyformationcontrol_BE} allows a final formation to be either (a) or (b).

\begin{figure}[t]
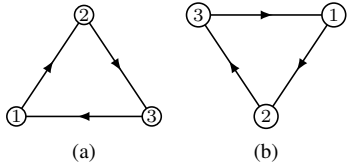

  \centering
  \def\myscale{0.45}
  \include{figure_tikz_undesiredEuqilibrium}
  \caption{The two formations are bearing equivalent but with opposite bearings.}
  \label{fig_twoKindsofTargetFormation}
\end{figure}

\subsection{Bearing-Based Formation Control Law}

The bearing-based formation control law considered in this paper is
\begin{align}\label{eq_controlLaw}
\dot{p}_i(t)= - \sum_{j\in\N_i} P_{g_{ij}^*}\left(p_i(t)-p_j(t)\right), \quad i\in\V,
\end{align}
where $P_{g_{ij}^*}=I_d-(g_{ij}^*) (g_{ij}^*)^\T$.
Control law \eqref{eq_controlLaw} was originally proposed in \cite{zhao2015ECC} for formations with undirected interaction topologies.
As will be shown in this paper later, it can also be applied to the directed case.
Several remarks on the control law are given below.
Firstly, control law \eqref{eq_controlLaw} is distributed because the control of each agent only requires the relative positions of its neighbors.
Secondly, the matrix $P_{g_{ij}^*}$ is an orthogonal projection matrix and hence control law \eqref{eq_controlLaw} has a clear geometric meaning.
As shown in Figure~\ref{fig_controlLawGeometricMeaning}, the control term $-P_{{g}_{ij}^*}({p}_i(t)-{p}_j(t))$ is the orthogonal projection of ${p}_j(t)-{p}_i(t)$ onto the orthogonal compliment of $g_{ij}$, and hence it acts to reduce the bearing error.

The matrix expression of control law \eqref{eq_controlLaw} is
\begin{align}\label{eq_controlLaw_matrix}
\dot{p}(t)=-\LB p(t),
\end{align}
where $\LB\in\R^{dn\times dn}$ and the $ij$th block submatrix of $\LB$ is
\begin{align*}
\left\{
  \begin{array}{ll}
      [\LB]_{ij}=0_{d\times d}, & i\ne j, (i,j)\notin\E, \\
      {[\LB]_{ij}}=-P_{g_{ij}^*}, & i\ne j, (i,j)\in\E, \\ 
      {[\LB]_{ii}}=\sum_{j\in\N_i}P_{g_{ij}^*}, & i\in\V. \\
  \end{array}
\right.
\end{align*}
The matrix $\LB$ can be interpreted as a matrix-weighted graph Laplacian.
It is jointly determined by the topological and Euclidean structure of the formation.
We call $\LB$ the \emph{bearing Laplacian} since it carries the information of both the bearings and the underlying graph of the formation.
The spectrum and null space of $\LB$ fully determine the stability of the formation dynamics \eqref{eq_controlLaw_matrix}.

\begin{figure}
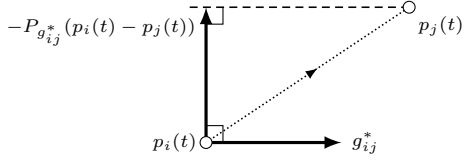

  \centering
  \def\myscale{0.45}
  \include{figure_tikz_ProtocolMeaning}
  \vspace{-10pt}
  \caption{The geometric interpretation of control law \eqref{eq_controlLaw}.}
  \label{fig_controlLawGeometricMeaning}
\end{figure}

\section{Bearing Persistence and Formation Stability}

In this section, we define the notion of bearing persistence by the bearing Laplacian and analyze the impact of the bearing persistence on the global formation stability of system \eqref{eq_controlLaw_matrix}.

We note that three subspaces play key roles in the formation stability analysis: $\Null(\LB)$, $\Null(\RB)$, and $\myspan\{\one\otimes I_d, p^*\}$.
First, $\Null(\LB)$ is the space where the solution $p(t)$ converges when the eigenvalues of $\LB$ are in the closed right half of the complex plane.
Second, $\Null(\RB)$ is the space of all the formations that are bearing equivalent to the target formation according to Theorem~\ref{result_bearingEquivalentCondition}.
The equality $\Null(\LB)=\Null(\RB)$ is required in order to solve Problem~\ref{problem_bearingonlyformationcontrol_BE}.
Third, $\myspan\{\one\otimes I_d, p^*\}$ is the space of all the formations that have the same shape as the target formation (but can have different translations and scales).
The equality $\Null(\LB)=\myspan\{\one\otimes I_d, p^*\}$ is desired when the formation is required to converge to a formation that has the same shape as the target formation.
We next explore the relationship between the three subspaces.

We first consider the simple and special case where the graph is undirected.
An undirected graph can be viewed as a special directed graph where $(i,j) \in \mathcal{E}\Leftrightarrow(j,i)\in \mathcal{E}$.
In the undirected case, the properties of the bearing Laplacian can be easily obtained as below.
The proof can be found in \cite{zhao2015NetLocalization}. For the sake of completeness, we also provide a proof below.

\begin{theorem}\label{result_LBRBNullSpace_undirected}
For an undirected formation $\G(p)$, the bearing Laplacian $\LB$ is symmetric positive semi-definite and satisfies
\begin{align*}
\Null(\LB)=\Null(\RB).
\end{align*}
\end{theorem}
\begin{proof}
Each undirected edge in the graph can be viewed as two directed edges with opposite directions.
Index the directed edges from $1$ to $m$.
Then $\LB$ can be written as $\LB=\bar{H}^\T \mydiag(P_{g_k})\bar{H}$, which is clearly symmetric and positive semi-definite.
Moreover, due to $P_{g_k}=P_{g_k}^\T P_{g_k}$, $\LB$ can be rewritten as
\begin{align*}
    \LB=\underbrace{\bar{H}^\T \mydiag{(P_{g_k}^\T )}}_{\tilde{R}_B^\T }\underbrace{\mydiag{(P_{g_k})}\bar{H}}_{\tilde{R}_B}.
\end{align*}
Note $\tilde{R}_B=\mydiag(\|e_k\|)\RB$ where $\RB$ is the bearing rigidity matrix.
As a result, $\Null(\tilde{R}_B)=\Null(\RB) = \Null(\LB)$.
\end{proof}

Based on Theorem~\ref{result_LBRBNullSpace_undirected}, the solution $p(t)$ to system \eqref{eq_controlLaw_matrix} in the undirected case converges into $\Null(\RB)$ and hence the finally converged formation is bearing equivalent to the target formation $\G(p^*)$.
In order to make $p(t)$ converge to $\myspan\{\one\otimes I_d, p^*\}$, it is required the formation to be infinitesimally bearing rigid according to Theorem~\ref{theorem_conditionInfiParaRigid}.

We now consider the general directed case.
When the graph is directed, the bearing Laplacian is not symmetric any more and we do \emph{not} have $\Null(\LB)=\Null(\RB)$ in general.

\begin{theorem}\label{result_nullLBNullRB}
For a directed formation $\G(p)$, the bearing Laplacian $\LB$ satisfies
\begin{align*}
\myspan\{\one\otimes I_d, p\}\subseteq\Null(\RB)\subseteq \Null(\LB).
\end{align*}
\end{theorem}
\begin{proof}
The fact $\myspan\{\one\otimes I_d, p\}\subseteq\Null(\RB)$ is given in Lemma~\ref{lemma_BearingRigidityProperty}(b).

We next prove $\Null(\RB)\subseteq \Null(\LB)$.
In order to do that, we firstly show that for any vector $p'\in\R^{dn}$ the equality $\RB p'=0$ holds if and only if
\begin{align}\label{eq_RBp=0Condition}
P_{g_{ij}}(p'_i-p'_j)=0, \quad \forall (i,j)\in\E.
\end{align}
As shown in Lemma~\ref{lemma_BearingRigidityProperty}, by indexing the directed edges from $1$ to $m$, the bearing rigidity matrix can be expressed by $\RB=\mydiag(P_{g_k}/\|e_k\|)\bar{H}$.
As a result, $\RB p'=0\Leftrightarrow \mydiag(P_{g_k}/\|e_k\|)\bar{H}p'=\mydiag(P_{g_k}/\|e_k\|)e'=0\Leftrightarrow P_{g_k}e_k'=0, \forall k\in\{1,\dots,m\}$. Therefore, $\RB p'=0$ if and only if \eqref{eq_RBp=0Condition} holds.
We secondly show that the equality $\LB p'=0$ holds if and only if
\begin{align}\label{eq_LBp=0Condition}
\sum_{j\in\N_i} P_{g_{ij}}(p'_i-p'_j)=0,\quad \forall i\in\V.
\end{align}
Since
\begin{align*}
\LB p'=
\left[
  \begin{array}{c}
    \vdots \\
    \sum_{j\in\N_i} P_{g_{ij}}(p'_i-p'_j) \\
    \vdots \\
  \end{array}
\right],
\end{align*}
we have $\LB p'=0$ if and only if \eqref{eq_LBp=0Condition} holds.
Finally, since equation \eqref{eq_RBp=0Condition} implies equation \eqref{eq_LBp=0Condition}, any vector in $\Null(\RB)$ is also in $\Null(\LB)$ and hence $\Null(\RB)\subseteq \Null(\LB)$.
\end{proof}

If $\Null(\LB)$ is larger than $\Null(\RB)$, the formation may converge to a final formation in $\Null(\LB)$ but not in $\Null(\RB)$, which is not bearing equivalent to the target formation.
Hence it is required that $\Null(\RB)=\Null(\LB)$ in order to solve Problem~\ref{problem_bearingonlyformationcontrol_BE}, which motivates the following notion.

\begin{definition}[Bearing Persistence]\label{definition_bearingPersistence}
A directed formation $\G(p)$ is \emph{bearing persistent} if $\Null(\RB)=\Null(\LB)$.
\end{definition}

The name of persistence actually is inherited from the notion of persistence in distance-based formation control problems \cite{Hendrickx2007IJRNC,Oh2015Automatica}.
Bearing persistence describes whether or not the directed graph is persistent with the bearing rigidity of a formation.
If they are not persistent, undesired equilibriums would appear for system \eqref{eq_controlLaw_matrix} and the target formation cannot be globally stabilized.
In addition, it follows from Theorem~\ref{result_LBRBNullSpace_undirected} that undirected formations are always bearing persistent.
The next result immediately follows from Theorem~\ref{result_nullLBNullRB} and Definition~\ref{definition_bearingPersistence}.

\begin{theorem}\label{theorem_NULLLB=1andp}
For a directed formation $\G(p)$, the equality $$\Null(\LB)=\myspan\{\one\otimes I_d, p\}$$ holds if and only if the formation is both infinitesimally bearing rigid and bearing persistent.
\end{theorem}
\begin{proof}
According to Theorem~\ref{result_nullLBNullRB}, the equality $\Null(\LB)=\myspan\{\one\otimes I_d, p\}$ holds if and only if $\Null(\LB)=\Null(\RB)$ and $\Null(\RB)=\myspan\{\one\otimes I_d, p\}$, where the former one holds if and only if the formation is bearing persistent by definition and the latter one holds if and only if the formation is infinitesimally bearing rigid by
Theorem~\ref{theorem_conditionInfiParaRigid}.
\end{proof}

Theorem~\ref{theorem_NULLLB=1andp} is important in the sense that it gives the necessary and sufficient condition for the equality $\Null(\LB)=\myspan\{\one\otimes I_d, p\}$ in the directed case.
This equality is desired when the formation is required to converge to a formation that has the same shape as the target formation.

\begin{figure}
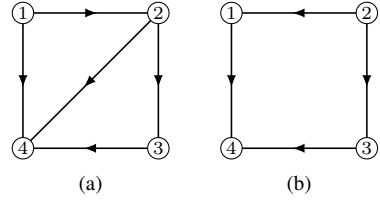

  \centering
  \def\myscale{0.45}
  \include{figure_tikz_examplePersistent}
\vspace{-5pt}
  \caption{Examples of bearing \emph{persistent} formations.}
  \label{fig_examplePersistent}
\end{figure}
\begin{figure}
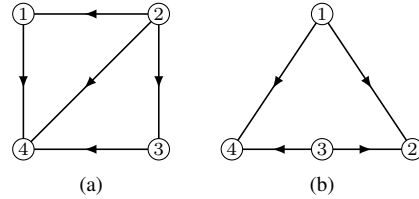

  \centering
  \def\myscale{0.45}
  \vspace{-10pt}
  \include{figure_tikz_exampleNonPersistent}
\vspace{-5pt}
  \caption{Examples of bearing \emph{non-persistent} formations}
  \label{fig_exampleNonPersistent}
\end{figure}

\subsection{Illustrative Examples}
Examples are shown in Figures~\ref{fig_examplePersistent} and \ref{fig_exampleNonPersistent} to illustrate the notion of bearing persistence.
We next specifically examine the two formations in Figure~\ref{fig_examplePersistent}(a) and Figure~\ref{fig_exampleNonPersistent}(a).
For the two formations, the agents have exactly the same positions as $p_1=[-1,1]^\T$, $p_2=[1,1]^\T$, $p_3=[1,-1]^\T$, and $p_4=[-1,-1]^\T$.
The underlying graphs of the two formations, however, are different; the edges between agents $1$ and $2$ have the opposite directions.
First of all, both of the two formations are infinitesimally bearing rigid because it can be verified that for both of the formations $\rank(\RB)=2n-3$. The null space of $\RB$ for the two formations are the same and given by
\begin{align*}
\Null(\RB)=\myspan\left\{
  \begin{array}{rrr}
    1 & 0 & -1 \\
    0 & 1 & 1\\
    1 & 0 & 1 \\
    0 & 1 & 1 \\
    1 & 0 & 1 \\
    0 & 1 & -1 \\
    1 & 0 & -1\\
    0 & 1 & -1 \\
  \end{array}
\right\}=\myspan\{\one\otimes I_2, p\}.
\end{align*}
For the formation in Figure~\ref{fig_examplePersistent}(a), it can be verified that $\Null(\LB)=\Null(\RB)$.
As a result, this formation is bearing persistent.
However, for the formation in Figure~\ref{fig_exampleNonPersistent}(a), the null space of $\LB$ is
\begin{align*}
\Null(\LB)=\myspan\left\{
  \begin{array}{rrrr}
    1 & 0 & -1 &0\\
    0 & 1 & 1&2\\
    1 & 0 & 1&-1 \\
    0 & 1 & 1&1 \\
    1 & 0 & 1&-2 \\
    0 & 1 & -1&0 \\
    1 & 0 & -1&0\\
    0 & 1 & -1&0 \\
  \end{array}
\right\}.
\end{align*}
As can be seen, in addition to the basis vectors in $\Null(\RB)$, an extra vector (the fourth) is contained in $\Null(\LB)$.
As a result, this formation is not bearing persistent.
We will later show by simulation that this non-persistent formation cannot be globally stabilized.

It must be noted that the bearing persistence of a formation is \emph{independent} to the bearing rigidity of the formation.
For example, the directed formation as shown in Figure~\ref{fig_examplePersistent}(a) is both bearing persistent and infinitesimally bearing rigid.
The one in Figure~\ref{fig_examplePersistent}(b) is bearing persistent but not bearing rigid.
The one in Figure~\ref{fig_exampleNonPersistent}(a) is infinitesimally bearing rigid but not bearing persistent.
The one in Figure~\ref{fig_exampleNonPersistent}(b) is neither infinitesimally bearing rigid nor bearing persistent.

\subsection{A Sufficient Condition for Bearing Persistence}

An important problem that follows the definition of bearing persistence is what kind of formations are bearing persistent.
We next give a sufficient condition for bearing persistence.
All the persistent formations given in Figure~\ref{fig_examplePersistent} can be examined by this sufficient condition.
For a directed edge $(i,j)\in\E$, since this edge leaves agent $i$ and enters agent $j$, we call it an \emph{outgoing} edge for agent $i$ and an \emph{incoming} edge for agent $j$.
In the formation dynamics \eqref{eq_controlLaw}, \emph{each agent is responsible for achieving all the bearings of its outgoing edges}.
It is intuitively reasonable that when the number of the outgoing edges of an agent is large, the chance that the agent is not able to archive all the desired edge bearings is high.
This intuition is consistent with the following sufficient condition for bearing persistence in $\R^2$.

\begin{proposition}\label{result_BP_sufficientCondition}
For a directed formation $\G(p)$ in $\R^2$, if each agent has at most two non-collinear outgoing edges, then the formation is bearing persistent.
\end{proposition}
\begin{proof}
We prove that in this case any vector $x\in\R^{2n}$ in $\Null(\LB)$ is also in $\Null(\RB)$.
Note $\LB x=0$ if and only if equation \eqref{eq_LBp=0Condition} holds.
Suppose agent $i\in\V$ has two outgoing edges (i.e., two neighbors).
Let the two neighbors be agents $k$ and $\ell$.
Then, equation \eqref{eq_LBp=0Condition} implies
\begin{align}\label{eq_LBp_twononcollinear}
P_{g_{ik}}(x_i-x_k)+P_{g_{i\ell}}(x_i-x_\ell)=0.
\end{align}
Note the orthogonal projection matrices $P_{g_{ik}}$ and $P_{g_{i\ell}}$ project vectors onto the orthogonal vectors of $g_{ik}$ and $g_{i\ell}$, respectively.
Since $g_{ik}$ and $g_{i\ell}$ are not collinear, equation \eqref{eq_LBp_twononcollinear} holds if and only if $P_{g_{ik}}(x_i-x_k)=0$ and $P_{g_{i\ell}}(x_i-x_\ell)=0$.
The case where agent $i$ has merely one outgoing edge (i.e., one neighbor) is trivial to analyze.
In summary, if each agent has at most two non-collinear outgoing edges, equation \eqref{eq_LBp=0Condition} implies \eqref{eq_RBp=0Condition} and consequently $\Null(\LB)=\Null(\RB)$.
\end{proof}

The sufficient condition in Proposition~\ref{result_BP_sufficientCondition} provides an intuitive way to check the bearing persistence of a formation in $\R^2$.
All the formations in Figure~\ref{fig_examplePersistent} can be concluded as bearing persistent by Proposition~\ref{result_BP_sufficientCondition}.
It should be noted that the outgoing edges in Proposition~\ref{result_BP_sufficientCondition} must be non-collinear; otherwise, the formation may be not bearing persistent. For example, although in the formation given in Figure~\ref{fig_exampleNonPersistent}(a) each agent has at most two outgoing edges, the formation is still non-persistent because the two outgoing edges of agent 3 are collinear.

\subsection{Discussion of the Formation Stability}

Up to now, we have explored the notion of bearing persistence and showed that the global stability cannot be guaranteed for formations that are not bearing persistent.
An important problem that has not been discussed is whether all the eigenvalues of the bearing Laplacian are in the closed right half of the complex plane.
For undirected formations, the answer is positive as shown in Theorem~\ref{result_LBRBNullSpace_undirected}.
For directed formations, the problem is nontrivial to solve and will be the subject of future work.
But according to a large amount of simulations, it is conjectured that the eigenvalues of the bearing Laplacian of a directed formation have nonnegative real parts.
If the conjecture is correct, then the formation dynamics \eqref{eq_controlLaw_matrix} is marginally stable.
Based on the linear system theory, the solution $p(t)$ converges to the projection of $p(0)$ onto $\Null(\LB)$.
If the target formation is bearing persistent, the final formation is in $\Null(\RB)$ and hence Problem~\ref{problem_bearingonlyformationcontrol_BE} can be solved.
Furthermore, if the target formation is both bearing persistent and infinitesimally bearing rigid (i.e., $\Null(\LB)=\myspan\{\one_n\otimes I_d\}$), the final formation has exactly the same shape (but maybe different translation and scale) as the desired target formation.

\section{Simulation Examples}

\begin{figure*}
  \centering
  \subfloat[Initial formation]{\includegraphics[width=0.2\linewidth]{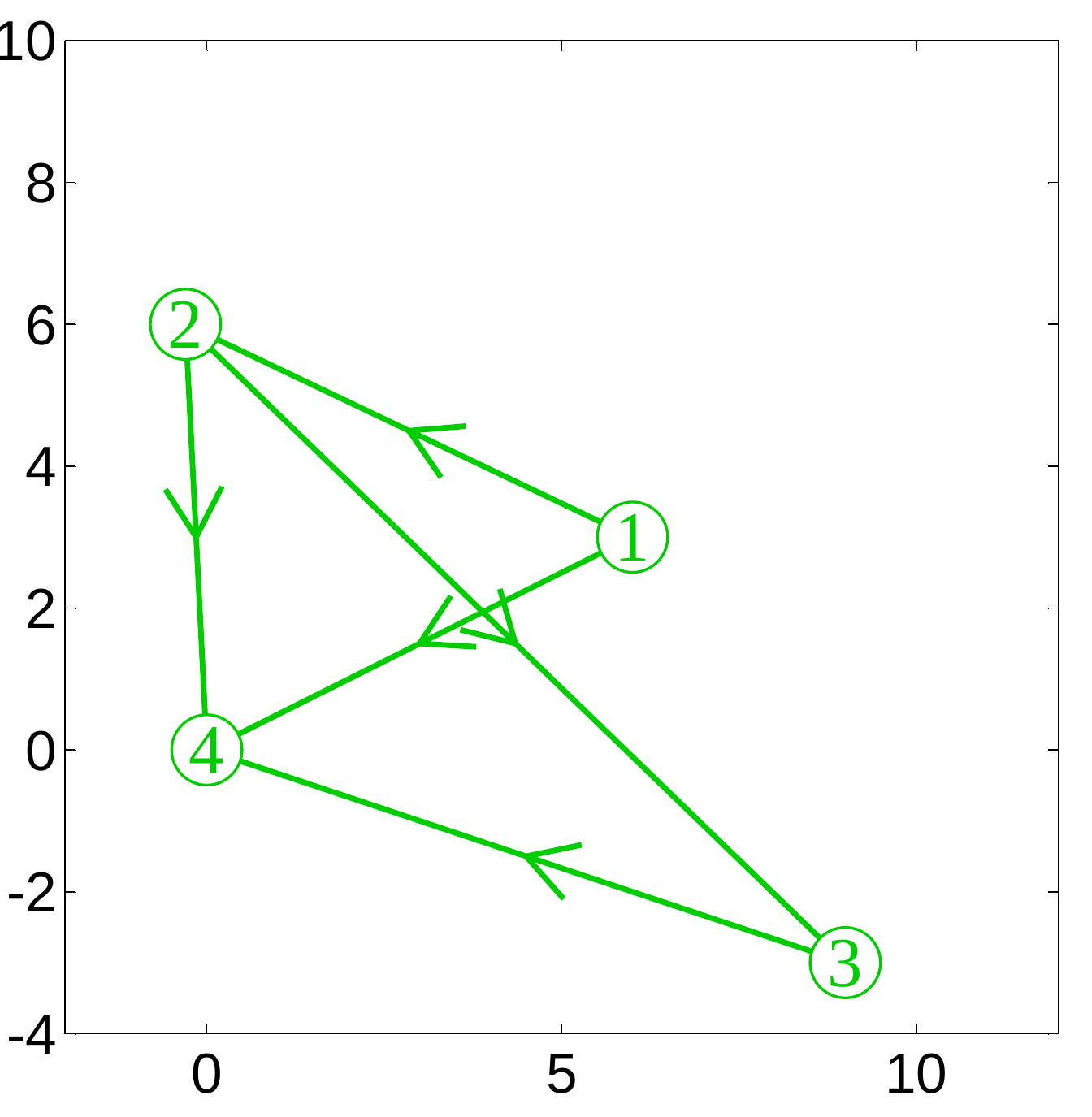}}
  \subfloat[Final formation]{\includegraphics[width=0.2\linewidth]{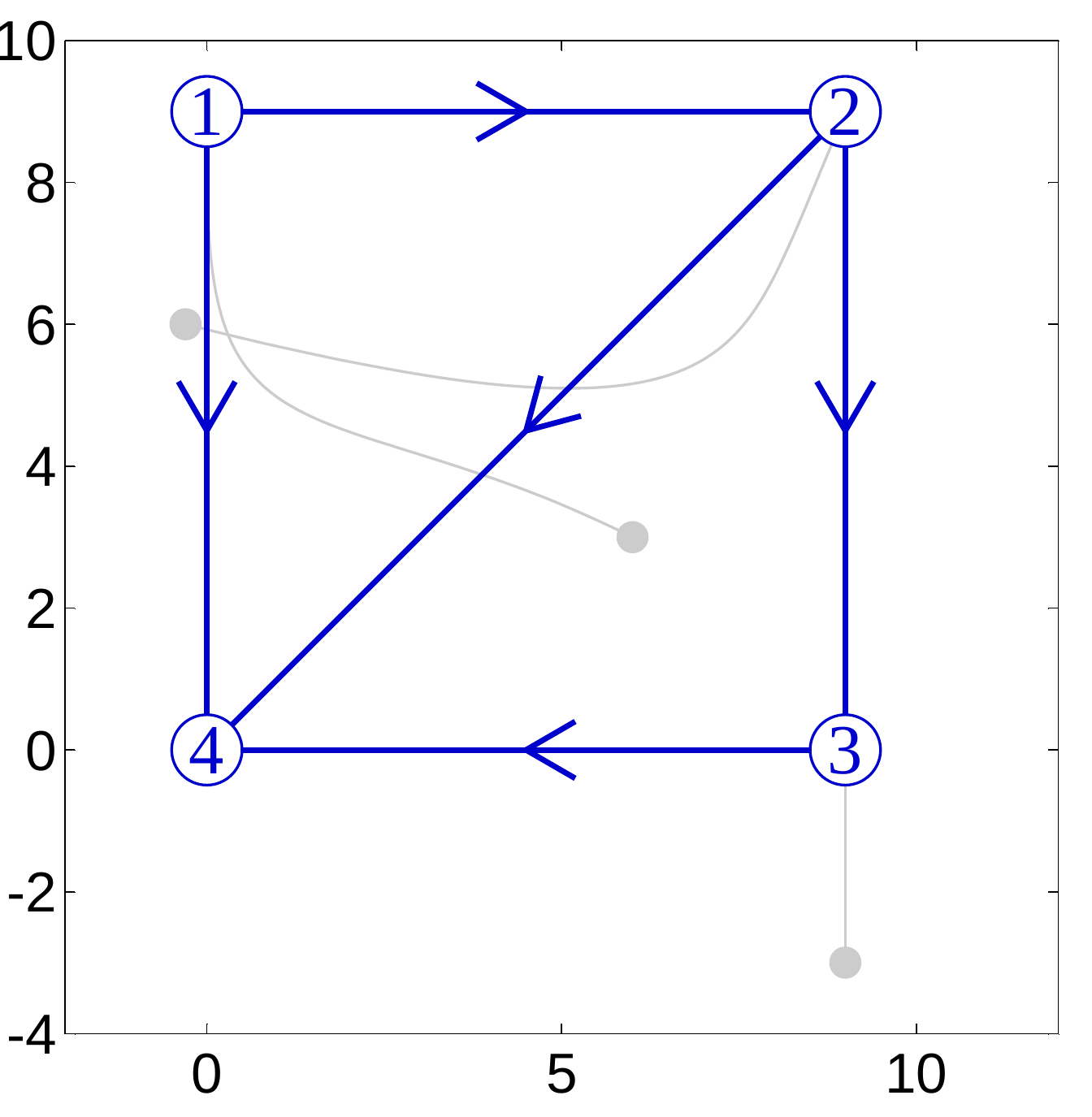}}
  \subfloat[$\|\dot{p}(t)\|=\|\LB p(t)\|$]{\includegraphics[width=0.25\linewidth]{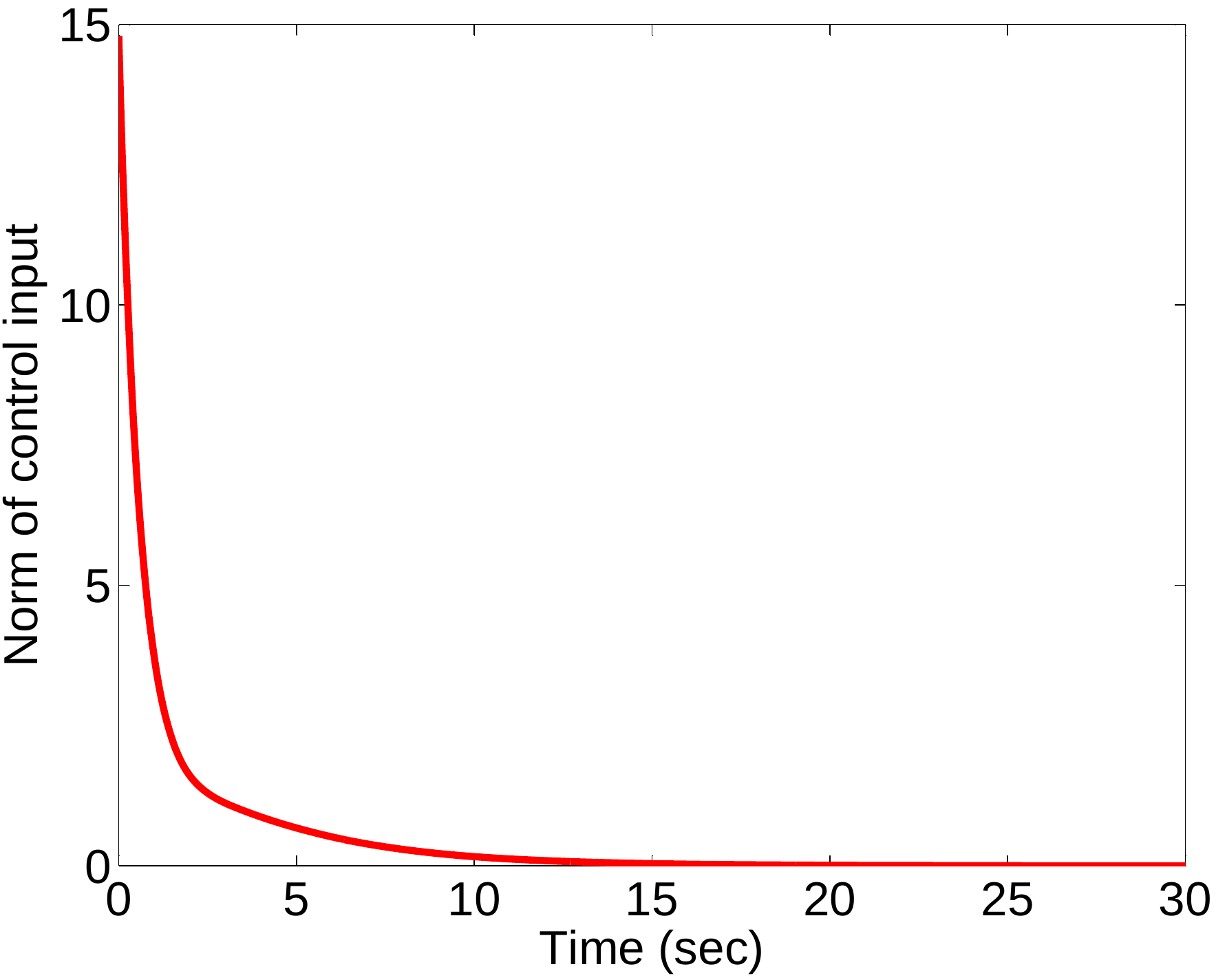}}
  \subfloat[$\sum_{(i,j)\in\E}\|P_{g_{ij}^*}(p_i(t)-p_j(t))\|$]{\includegraphics[width=0.25\linewidth]{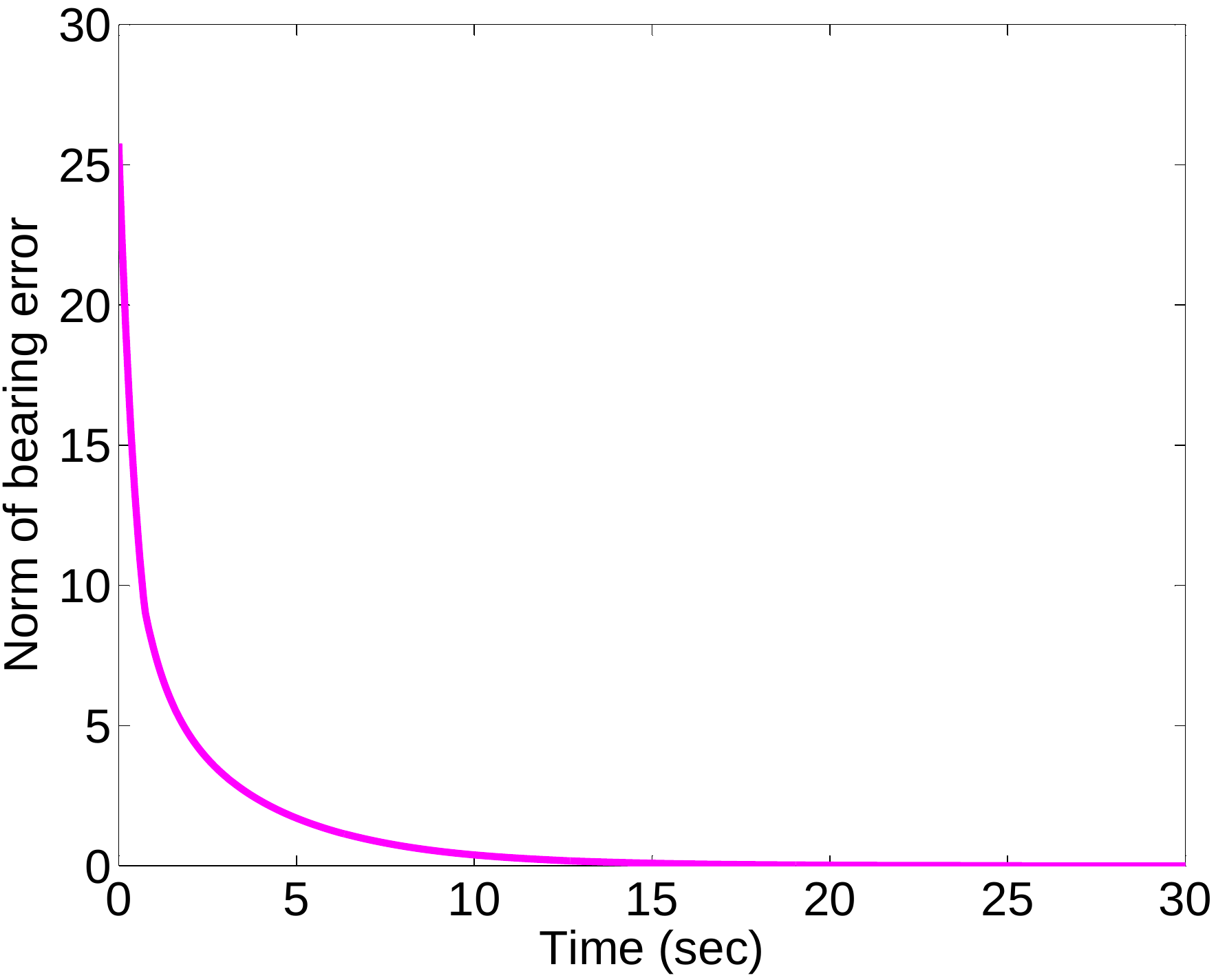}}\\
  \caption{The simulation result for the \emph{bearing persistent} target formation as shown in Figure~\ref{fig_examplePersistent}(a).}
  \label{fig_sim_persistent}
\end{figure*}

\begin{figure*}
  \centering
  \subfloat[Initial formation]{\includegraphics[width=0.2\linewidth]{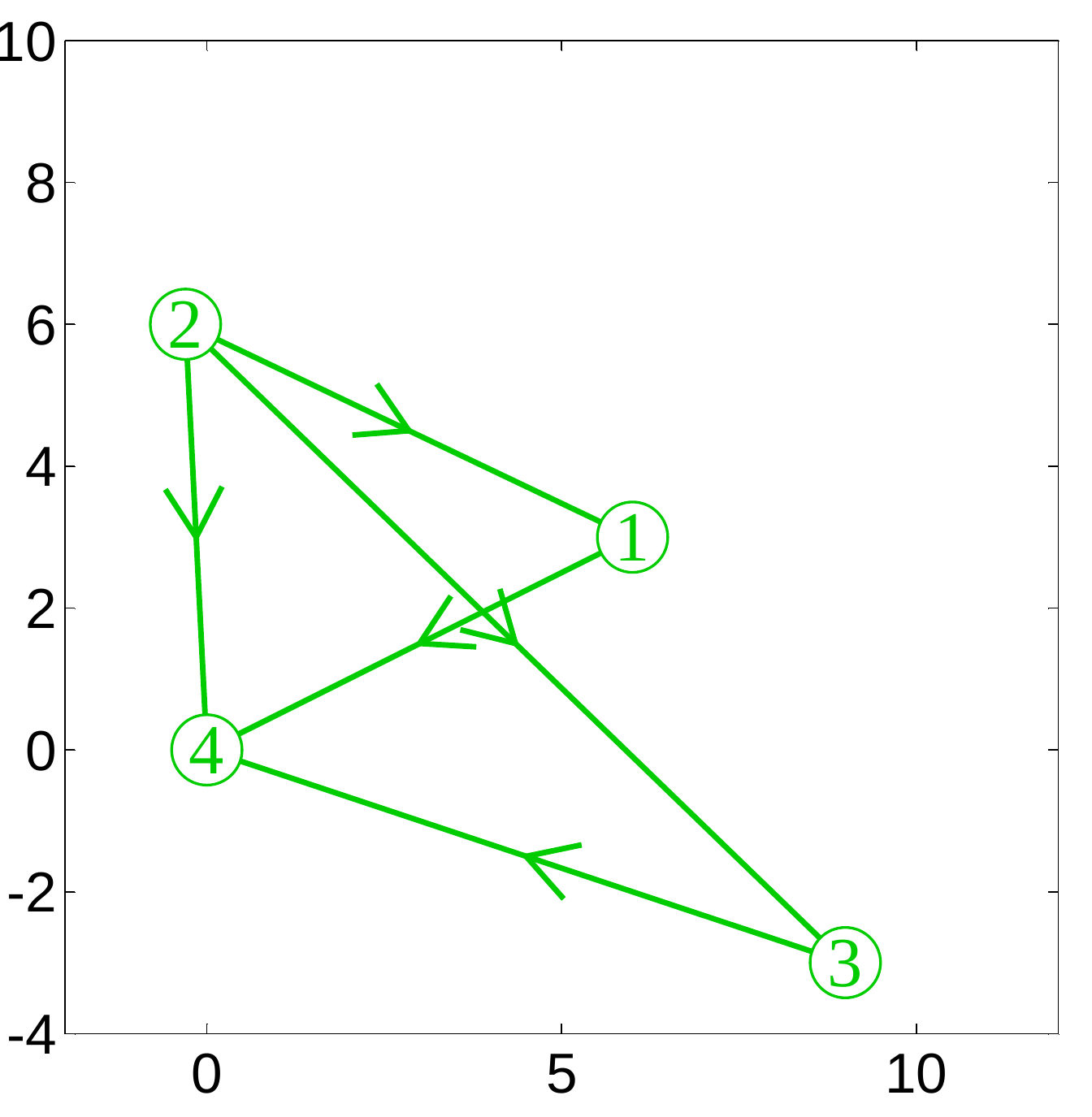}}
  \subfloat[Final formation]{\includegraphics[width=0.2\linewidth]{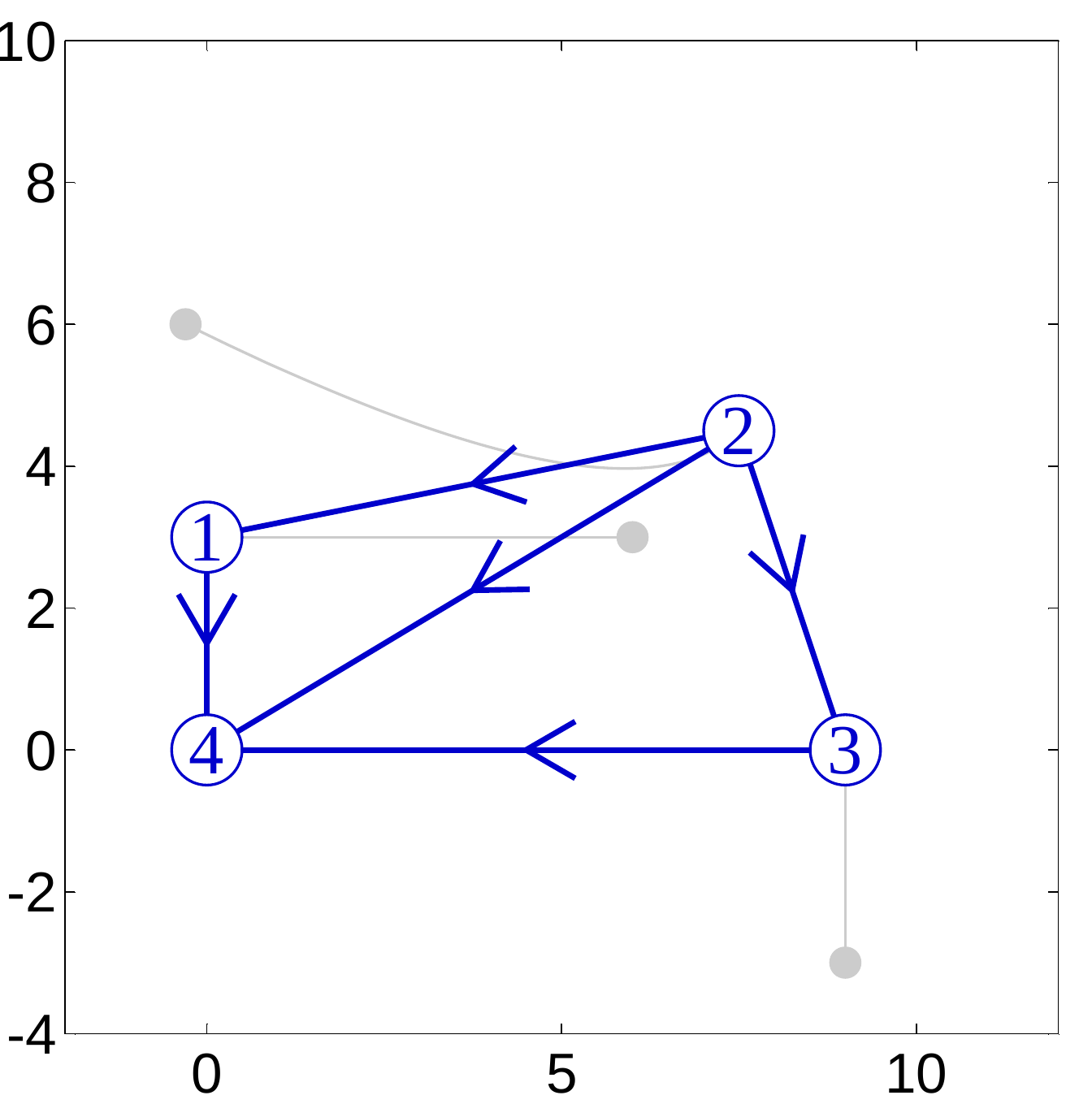}}
  \subfloat[$\|\dot{p}(t)\|=\|\LB p(t)\|$]{\includegraphics[width=0.25\linewidth]{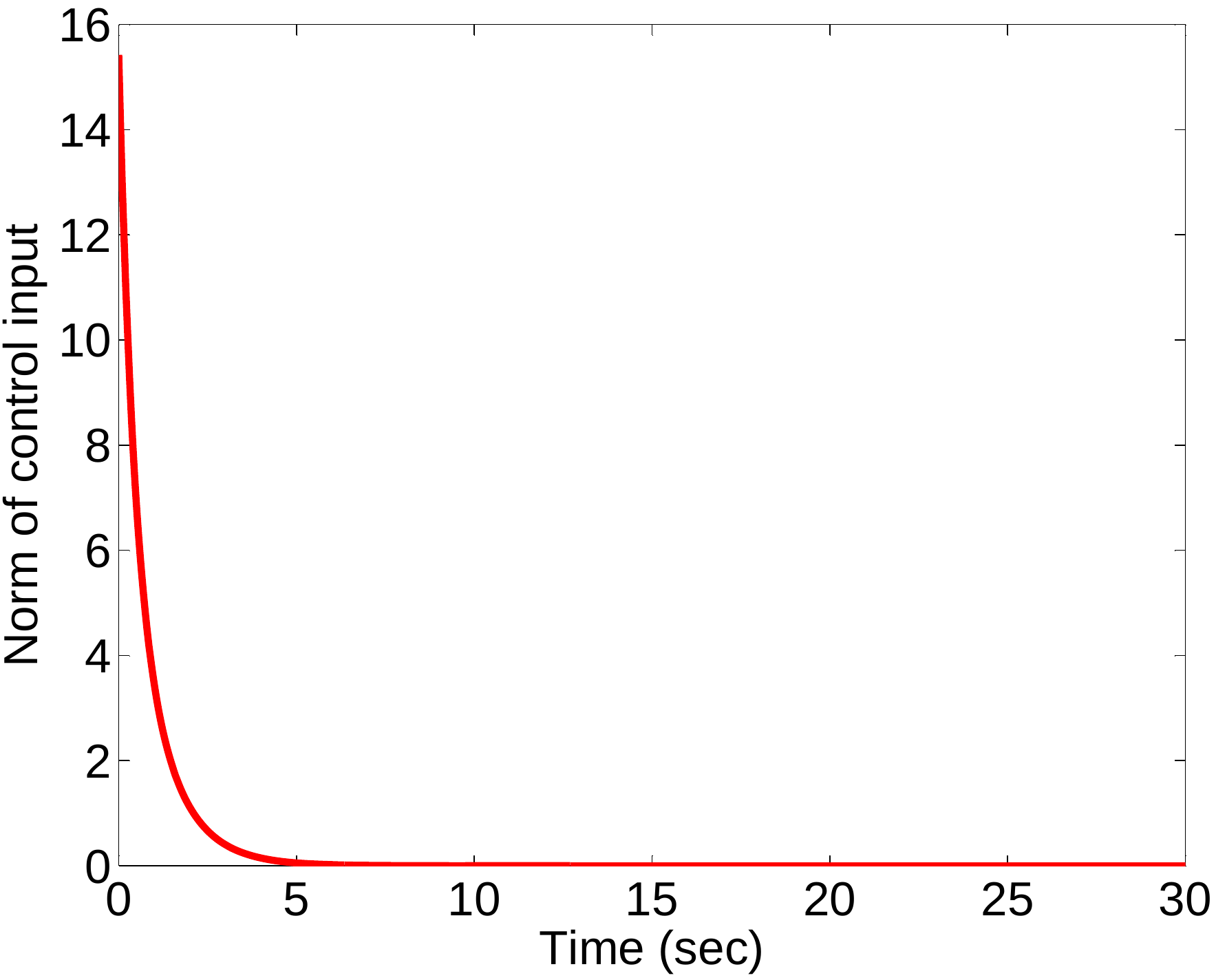}}
  \subfloat[$\sum_{(i,j)\in\E}\|P_{g_{ij}^*}(p_i(t)-p_j(t))\|$]{\includegraphics[width=0.25\linewidth]{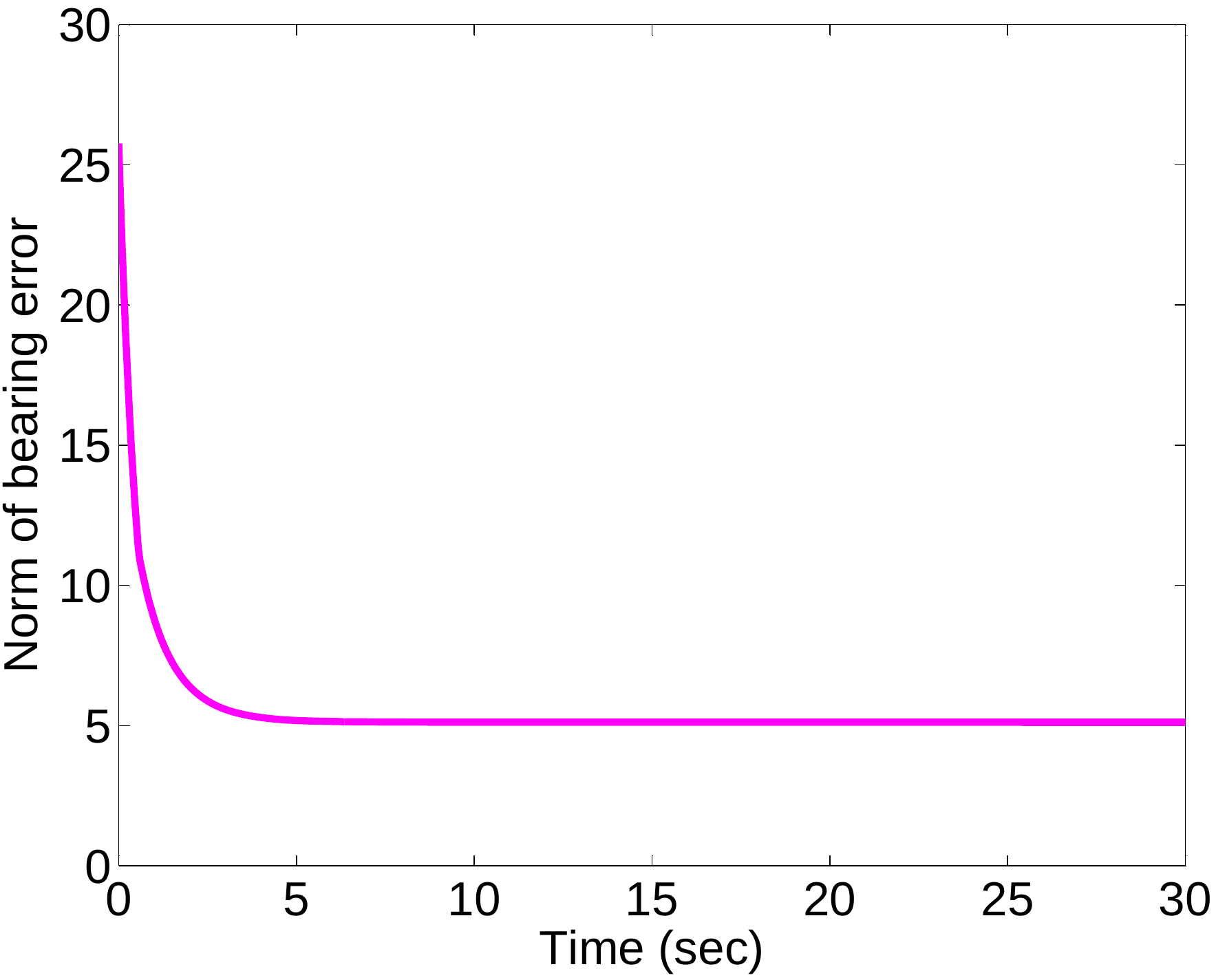}}\\
  \caption{The simulation result for the \emph{bearing non-persistent} target formation as shown in Figure~\ref{fig_exampleNonPersistent}(a).}
  \label{fig_sim_nonpersistent}
\end{figure*}

We firstly consider the \emph{bearing persistent} target formation given in Figure~\ref{fig_examplePersistent}(a).
The simulation results are given in Figure~\ref{fig_sim_persistent}.
As can be seen, the control law \eqref{eq_controlLaw} can steer the formation to a final formation that is bearing equivalent to the target formation.
We secondly consider the target formation given in Figure~\ref{fig_exampleNonPersistent}(a) that is \emph{not bearing persistent}.
The simulation results are given in Figure~\ref{fig_sim_nonpersistent}.
As can be seen, the control law \eqref{eq_controlLaw} steers the formation to an undesired formation.
That is because agent 2 is not able to find a position to satisfy all the three bearing constraints $g_{21}^*$, $g_{24}^*$, and $g_{23}^*$.
The control input $\|\dot{p}(t)\|=\|\LB p(t)\|$ converges to zero, which means $p(t)$ converges into $\Null(\LB)$, but as shown in Figure~\ref{fig_sim_nonpersistent}(d) the bearing error $\sum_{(i,j)\in\E}\|P_{g_{ij}^*}(p_i(t)-p_j(t))\|$ can \emph{not} converge to zero, which means $p(t)$ does not converge to $\Null(\RB)$.
As a result, $p(t)$ finally converges into $\Null(\LB)\setminus\Null(\RB)$.

\section{Conclusions and Future Work}

In this paper, we studied bearing-based formation stabilization with directed interaction topologies.
We studied a linear distributed control law that can stabilize bearing-constrained target formations in arbitrary dimensions.
It was shown that an important notion, bearing persistence, emerges and critically affects the formation stability in the directed case.
When a target formation is not bearing persistent, undesired equilibriums will appear and the global formation stability cannot be guaranteed.
A sufficient condition that provides an intuitive way to examine the bearing persistence of two-dimensional formations was also proposed.

The work presented in this paper is a first step towards solving the general problem of bearing-based formation control.
Several important problems are not completely solved in this paper and should be studied in the future.
First, the spectrum of the bearing Laplacian with directed graphs is not solved by this paper.
Second, although a sufficient condition for the bearing persistence of formations in the two-dimensional space was proposed, the necessary and sufficient conditions have not been obtained.
Third, the left null space of the bearing Laplacian is also important for determining the final converged value of the formation and should be studied.
Finally, this paper only considered the leaderless scheme. In order to control the centroid and scale as well as the maneuver of the formation, the leader-follower scheme should be studied in the future.

{\small
\section*{Acknowledgements}
The work presented here has been supported by the Israel Science Foundation (grant No. 1490/1).

\bibliography{myOwnPub,zsyReferenceAll} 
\bibliographystyle{ieeetr}
}
\end{document}

%% file: before_document.tex
\usepackage{graphicx}
\usepackage{subfig}
\usepackage{amsmath}
\usepackage{amsthm} 
\usepackage{amssymb}
\usepackage[font=footnotesize]{caption} 
\usepackage{subfig} 
\usepackage[noadjust]{cite} 
\usepackage{color}
\usepackage{algorithm} 
\usepackage{algpseudocode} 
\usepackage{enumerate}
\usepackage[hidelinks,colorlinks=false]{hyperref}
\usepackage{setspace}

\usepackage{tikz}
\usetikzlibrary{calc} 
\usetikzlibrary{shapes} 
\usetikzlibrary{chains}
\usetikzlibrary{fit}
\usetikzlibrary{arrows}
\usetikzlibrary{decorations.text} 
\usetikzlibrary{decorations.markings}

\newtheorem{theorem}{Theorem}
\newtheorem{lemma}{Lemma}

\newtheorem{proposition}{Proposition}

\newtheorem{definition}{Definition}
\newtheorem{problem}{Problem}

\newcommand{\Null}{\mathrm{Null}}
\newcommand{\Range}{\mathrm{Range}}
\newcommand{\one}{\mathbf{1}}
\newcommand{\rank}{\mathrm{rank}}
\newcommand{\myspan}{\mathrm{span}}
\newcommand{\mydiag}{\mathrm{diag}}

\newcommand{\blkdiag}{\mathrm{blkdiag}}

\newcommand{\T}{\mathrm{T}}

\newcommand{\R}{\mathbb{R}}
\newcommand{\G}{\mathcal{G}}
\newcommand{\E}{\mathcal{E}}
\newcommand{\V}{\mathcal{V}}
\newcommand{\N}{\mathcal{N}}

\newcommand{\RB}{R_{\text{\scriptsize{B}}}}
\newcommand{\FB}{F_{\text{\scriptsize{B}}}}
\newcommand{\LB}{L_{\text{\scriptsize{B}}}}

\graphicspath{{figures/}}

%% file: figure_tikz_undesiredEuqilibrium.tex
\tikzset{->-/.style={decoration={
  markings,
  mark=at position #1 with {\arrow{>}}},postaction={decorate}}
}
\subfloat[]{
\begin{tikzpicture}[scale=\myscale]
\def\length{2}
\def\linetype{semithick}
\coordinate (origin) at (0,0);
\coordinate (x1) at (-\length,-\length/2);
\coordinate (x2) at (0,\length);
\coordinate (x3) at (\length,-\length/2);
\coordinate (x1_ref) at (\length,\length/2);
\coordinate (x2_ref) at (0,-\length);
\coordinate (x3_ref) at (-\length,\length/2);
\def\radius{8pt}
\def\mycolor{black}
\draw [\linetype,->-=0.55,>=latex,draw=\mycolor] (x1)--(x2);
\draw [\linetype,->-=0.55,>=latex,draw=\mycolor] (x2)--(x3);
\draw [\linetype,->-=0.55,>=latex,draw=\mycolor] (x3)--(x1);
\draw [\linetype,fill=white,draw=\mycolor](x1) circle [radius=\radius];
\draw (x1) node[] {\scriptsize{\color{\mycolor}$1$}};
\draw [\linetype,fill=white, draw=\mycolor](x2) circle [radius=\radius];
\draw (x2) node[] {\scriptsize{\color{\mycolor}$2$}};
\draw [\linetype,fill=white, draw=\mycolor](x3) circle [radius=\radius];
\draw (x3) node[] {\scriptsize{\color{\mycolor}$3$}};
\end{tikzpicture}
}
\subfloat[]{
\begin{tikzpicture}[scale=\myscale]
\def\length{2}
\def\linetype{semithick}
\coordinate (origin) at (0,0);
\coordinate (x1) at (-\length,-\length/2);
\coordinate (x2) at (0,\length);
\coordinate (x3) at (\length,-\length/2);
\coordinate (x1_ref) at (\length,\length/2);
\coordinate (x2_ref) at (0,-\length);
\coordinate (x3_ref) at (-\length,\length/2);
\def\radius{10pt}
\def\mycolor2{black}
\draw [\linetype,\mycolor2,->-=0.55,>=latex] (x1_ref)--(x2_ref);
\draw [\linetype,\mycolor2,->-=0.55,>=latex] (x2_ref)--(x3_ref);
\draw [\linetype,\mycolor2,->-=0.55,>=latex] (x3_ref)--(x1_ref);
\draw [\linetype,fill=white,draw=\mycolor2](x1_ref) circle [radius=\radius];
\draw (x1_ref) node[] {\color{\mycolor2}\scriptsize{$1$}};
\draw [\linetype,fill=white,draw=\mycolor2](x2_ref) circle [radius=\radius];
\draw (x2_ref) node[] {\color{\mycolor2}\scriptsize{$2$}};
\draw [\linetype,fill=white,draw=\mycolor2](x3_ref) circle [radius=\radius];
\draw (x3_ref) node[] {\color{\mycolor2}\scriptsize{$3$}};
\end{tikzpicture}
}

%% file: figure_tikz_ProtocolMeaning.tex
\tikzset{->-/.style={decoration={
  markings,
  mark=at position #1 with {\arrow{>}}},postaction={decorate}}
}

\begin{tikzpicture}[scale=\myscale]
            \coordinate (xi) at (0,0);
            \coordinate (xj) at (6,4);
            \coordinate (proj) at (0,4);
            \coordinate (projLeft) at (-0.5,4);
            \def\unit{4}
            \coordinate (gij)                 at (\unit,0);
            \draw [densely dotted,->-=0.55,>=latex, semithick] (xi)--(xj);
            \draw [densely dashed, semithick] (xj)--(projLeft);
            \draw[->, >=latex, very thick] (xi) -- (gij) node [right] {\scriptsize{$g_{ij}^*$}};
            \draw[->, >=latex, very thick] (xi) -- (proj) node [below left] {\scriptsize{$-P_{{g}_{ij}^*}({p}_i(t)-{p}_j(t))$}};
            \def\radius{5pt}
            \draw [fill=white](xi) circle [radius=\radius];
            \draw [fill=white](xj) circle [radius=\radius];
            \draw (xi) node[left] {\scriptsize{${p}_i(t)$}};
            \draw (xj) node[below right] {\scriptsize{${p}_j(t)$}};
            \def\dis{0.5}
            \coordinate (p1) at (0,\dis);
            \coordinate (p2) at (\dis,\dis);
            \coordinate (p3) at (\dis,0);
            \draw [] (p1)--(p2)--(p3);
            \coordinate (p4) at (0,4-\dis);
            \coordinate (p5) at (\dis,4-\dis);
            \coordinate (p6) at (\dis,4);
            \draw [] (p4)--(p5)--(p6);
\end{tikzpicture}

%% file: figure_tikz_examplePersistent.tex
\tikzset{->-/.style={decoration={
  markings,
  mark=at position #1 with {\arrow{>}}},postaction={decorate}}
}
\def\mycolor{black}
\def\radius{9pt}
\def\length{4}
\def\linetype{} 
\subfloat[]{
\begin{tikzpicture}[scale=\myscale]
\coordinate (x1) at (0,\length);
\coordinate (x2) at (\length,\length);
\coordinate (x3) at (\length,0);
\coordinate (x4) at (0,0);
\draw [\linetype,semithick,->-=0.55,>=latex,draw=\mycolor] (x1)--(x4);
\draw [\linetype,semithick,->-=0.55,>=latex,draw=\mycolor] (x1)--(x2);
\draw [\linetype,semithick,->-=0.55,>=latex,draw=\mycolor] (x2)--(x3);
\draw [\linetype,semithick,->-=0.55,>=latex,draw=\mycolor] (x3)--(x4);
\draw [\linetype,semithick,->-=0.55,>=latex,draw=\mycolor] (x2)--(x4);
\draw [fill=white](x1) circle [radius=\radius];
\draw (x1) node[] {\scriptsize{\color{\mycolor}$1$}};
\draw [fill=white](x2) circle [radius=\radius];
\draw (x2) node[] {\scriptsize{\color{\mycolor}$2$}};
\draw [fill=white](x3) circle [radius=\radius];
\draw (x3) node[] {\scriptsize{\color{\mycolor}$3$}};
\draw [fill=white](x4) circle [radius=\radius];
\draw (x4) node[] {\scriptsize{\color{\mycolor}$4$}};
\end{tikzpicture}
}
\quad
\subfloat[]{
\begin{tikzpicture}[scale=\myscale]
\coordinate (x1) at (0,\length);
\coordinate (x2) at (\length,\length);
\coordinate (x3) at (\length,0);
\coordinate (x4) at (0,0);
\draw [\linetype,semithick,->-=0.55,>=latex,draw=\mycolor] (x1)--(x4);
\draw [\linetype,semithick,->-=0.55,>=latex,draw=\mycolor] (x2)--(x1);
\draw [\linetype,semithick,->-=0.55,>=latex,draw=\mycolor] (x2)--(x3);
\draw [\linetype,semithick,->-=0.55,>=latex,draw=\mycolor] (x3)--(x4);
\draw [fill=white](x1) circle [radius=\radius];
\draw (x1) node[] {\scriptsize{\color{\mycolor}$1$}};
\draw [fill=white](x2) circle [radius=\radius];
\draw (x2) node[] {\scriptsize{\color{\mycolor}$2$}};
\draw [fill=white](x3) circle [radius=\radius];
\draw (x3) node[] {\scriptsize{\color{\mycolor}$3$}};
\draw [fill=white](x4) circle [radius=\radius];
\draw (x4) node[] {\scriptsize{\color{\mycolor}$4$}};
\end{tikzpicture}
}

%% file: figure_tikz_exampleNonPersistent.tex
\tikzset{->-/.style={decoration={
  markings,
  mark=at position #1 with {\arrow{>}}},postaction={decorate}}
}
\def\mycolor{black}
\def\radius{9pt}
\def\length{4}
\def\linetype{} 
\subfloat[]{
\begin{tikzpicture}[scale=\myscale]
\coordinate (x1) at (0,\length);
\coordinate (x2) at (\length,\length);
\coordinate (x3) at (\length,0);
\coordinate (x4) at (0,0);
\draw [\linetype,semithick,->-=0.55,>=latex,draw=\mycolor] (x1)--(x4);
\draw [\linetype,semithick,->-=0.55,>=latex,draw=\mycolor] (x2)--(x1);
\draw [\linetype,semithick,->-=0.55,>=latex,draw=\mycolor] (x2)--(x3);
\draw [\linetype,semithick,->-=0.55,>=latex,draw=\mycolor] (x3)--(x4);
\draw [\linetype,semithick,->-=0.55,>=latex,draw=\mycolor] (x2)--(x4);
\draw [fill=white](x1) circle [radius=\radius];
\draw (x1) node[] {\scriptsize{\color{\mycolor}$1$}};
\draw [fill=white](x2) circle [radius=\radius];
\draw (x2) node[] {\scriptsize{\color{\mycolor}$2$}};
\draw [fill=white](x3) circle [radius=\radius];
\draw (x3) node[] {\scriptsize{\color{\mycolor}$3$}};
\draw [fill=white](x4) circle [radius=\radius];
\draw (x4) node[] {\scriptsize{\color{\mycolor}$4$}};
\end{tikzpicture}
}
\quad
\subfloat[]{
\begin{tikzpicture}[scale=\myscale]
\coordinate (x1) at (0,\length);
\coordinate (x2) at (\length/3*2,0);
\coordinate (x3) at (0,0);
\coordinate (x4) at (-\length/3*2,0);
\draw [\linetype,semithick,->-=0.55,>=latex,draw=\mycolor] (x1)--(x4);
\draw [\linetype,semithick,->-=0.55,>=latex,draw=\mycolor] (x1)--(x2);
\draw [\linetype,semithick,->-=0.55,>=latex,draw=\mycolor] (x3)--(x2);
\draw [\linetype,semithick,->-=0.55,>=latex,draw=\mycolor] (x3)--(x4);
\draw [fill=white](x1) circle [radius=\radius];
\draw (x1) node[] {\scriptsize{\color{\mycolor}$1$}};
\draw [fill=white](x2) circle [radius=\radius];
\draw (x2) node[] {\scriptsize{\color{\mycolor}$2$}};
\draw [fill=white](x3) circle [radius=\radius];
\draw (x3) node[] {\scriptsize{\color{\mycolor}$3$}};
\draw [fill=white](x4) circle [radius=\radius];
\draw (x4) node[] {\scriptsize{\color{\mycolor}$4$}};
\end{tikzpicture}
}